\documentclass[11pt,reqno]{amsart}

\usepackage{amssymb,amsmath,amsthm,amscd,latexsym,amsfonts}
\usepackage{mathtools}
\usepackage[T1]{fontenc}
\usepackage{graphicx}
\usepackage{xcolor}
\usepackage{cite}
\usepackage{float}
\usepackage[a4paper,top=3cm, bottom=3cm, left=3cm, right=3cm]{geometry}

\newtheorem{thm}{Theorem}

\newtheorem{lemma}{Lemma}

\newtheorem{rk}{Remark}

\numberwithin{equation}{section} 

\newcommand{\bea}{\begin{eqnarray}}
	\newcommand{\eea}{\end{eqnarray}}

\begin{document}
	\title [Bubble coalescence in interacting system of DNA]
	{Bubble coalescence in interacting system of DNA molecules}
	
	\author {U.A. Rozikov}
			 
	\address{ U.Rozikov$^{a,b,c}$\begin{itemize}
			\item[$^a$] V.I.Romanovskiy Institute of Mathematics,  9, Universitet str., 100174, Tashkent, Uzbekistan;
			\item[$^b$] AKFA University, National Park Street, Barkamol MFY,
			Mirzo-Ulugbek district, Tashkent, Uzbekistan;
			\item[$^c$] National University of Uzbekistan,  4, Universitet str., 100174, Tashkent, Uzbekistan.
	\end{itemize}}
	\email{rozikovu@yandex.ru}
	
	\begin{abstract}
We consider two models of interacting DNA molecules: \textit{First} is (four parametric)  bubble coalescence model in interacting DNAs (shortly: BCI-DNA). 
\textit{Second} is (three parametric)  bubble coalescence model in a condensed DNA molecules (shortly BCC-DNA).  
 
To study bubble coalescence thermodynamics of BCI-DNA and BCC-DNA models we use methods of statistical physics.  Namely, we define  Hamiltonian of each model and give their translation-invariant Gibbs measures (TIGMs). 

For the first model we find parameters such that corresponding Hamiltonian has up to three  TIGMs (three phases of system)  biologically meaning existence of three states: ``No bubble coalescence'',  ``Dominated soft zone'', ``Bubble coalescence''.

For the second model we show that for any (admissible) parameters this model has unique TIGM. This is a state where ``No bubble coalescence'' phase dominates. 
\end{abstract}
\maketitle



\section{Introduction}
It is known that \cite{book} each molecule of DNA is a double helix formed from two complementary strands of nucleotides
held together by hydrogen bonds between $G+C$ and $A+T$ base pairs, where  $C$=cytosine, $G$=guanine, $A$=adenine,
and $T$=thymine. 

Following \cite{FM}, \cite{Me} (see also references therein) we note that under physiological conditions the double helix is the equilibrium structure of DNA, its
stability controlled by hydrogen bonding of base pairs
 and stacking between these pairs. By change 
of temperature ($T>0$) double-stranded DNA progressively denatures, yielding regions of single-stranded DNA
(DNA bubbles) consisting broken base pairs. Consequently, the double strand fully denatures, the helix-coil transition at the melting temperature $T_m$.
Fueled by thermal activation, DNA bubbles occur spontaneously and fluctuate in size until closure ($T<T_m$) or
denaturation ($T>T_m$). This DNA breathing can be interpreted as a random walk in the one-dimensional coordinate $x$, the number of denatured base pairs, when one assumes that base pair unzipping and zipping
occur on a slower time scale than the relaxation of the
polymeric degrees of freedom of the bubbles. 

Investigation of DNA breathing (the bubble
dynamics) is motivated by providing a test case for new methods in statistical mechanical systems, where the dynamics 
of DNA bubbles can be probed on the single molecule level
in real time.

In \cite{FM} the authors showed that the
fluctuation dynamics of DNA denaturation bubbles can be
mapped onto the imaginary time Schr\"odinger equation of
the quantum Coulomb problem, allowing to calculate
the bubble lifetime distributions and associated correlation
functions depending on the temperature.  

In \cite{NPA} the authors studied the coalescence of two DNA-bubbles initially located at weak segments and separated by a more stable barrier region in a designed construct of double-stranded DNA. Moreover, the bubble dynamics is mapped on the problem of two vicious walkers in opposite potentials.
In Fig.\ref{bub} a schematic version of this model is given. 

\begin{figure}
	\includegraphics[width=14cm]{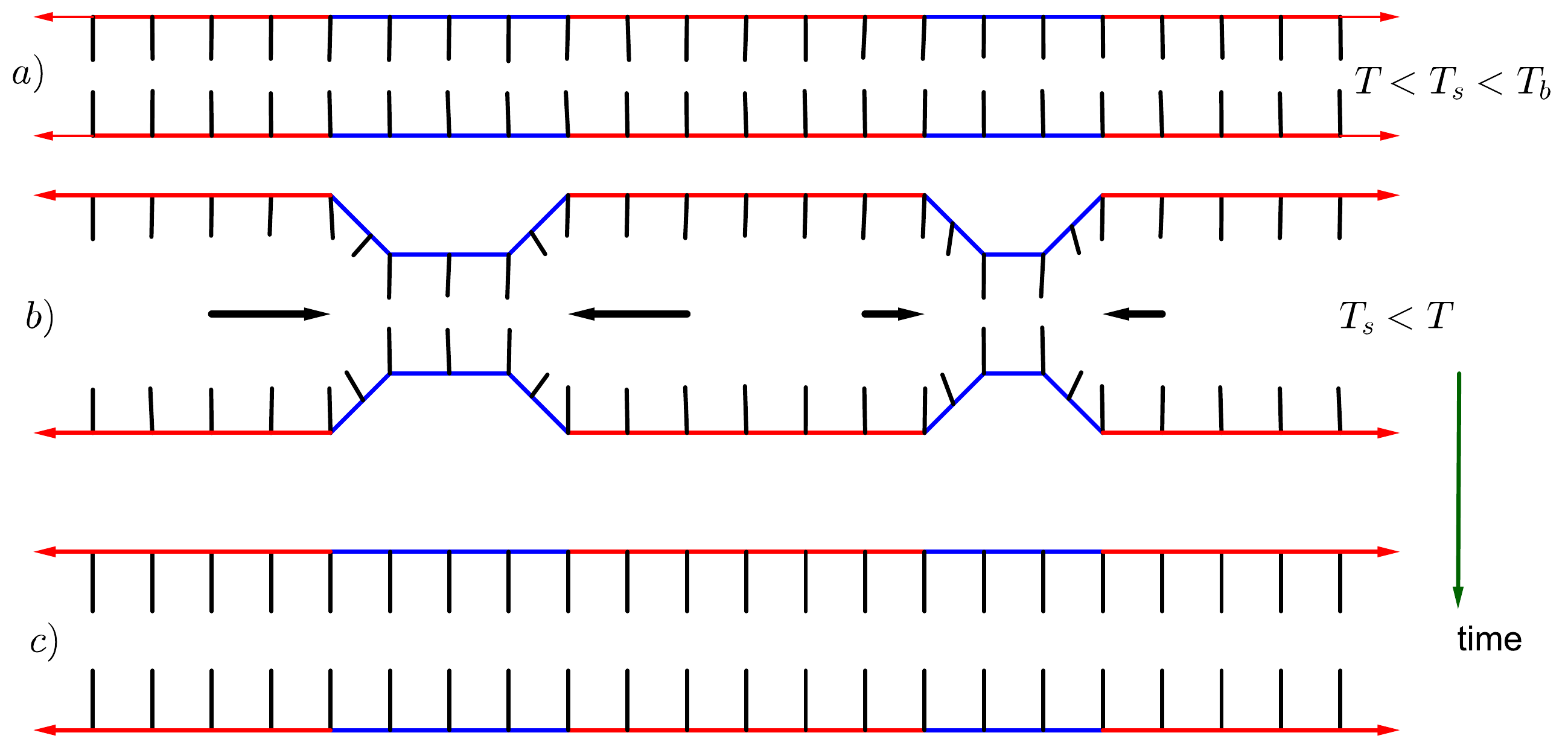}\\
	\caption{A schematic of the bubble coalescence setup in a DNA molecule. a) All base pairs closed when $T<T_s<T_b$, where $T_s$ melting temperature of soft zones (red) and $T_b$ is the melting temperature of barrier zones (blue). b) Soft zones open by raising the temperature above $T_s$. Successive
		opening of the barrier driven mainly by fluctuations $(T<T_b)$ or drift $(T>T_b)$ until coalescence. The state of DNA at a region is defined as the positions of  interfaces between the closed and broken base pairs. c) The state when barriers are removed (occupied by soft zones).}\label{bub}
\end{figure}

The structure of DNA can be described using
methods of statistical physics (see \cite{Sw}, \cite{T}). This study makes an important connection between the structure of DNA sequence and {\it temperature}; e.g., phase transitions in such a system may be interpreted as a conformational restructuring.

Fig.\ref{bub} shows the bubble coalescence in a \textbf{single} DNA molecule. But DNA as a polymer has physical properties,\footnote{https://bionano.physics.illinois.edu/node/203} under the right conditions, DNA molecules attract and condense into a compact state. The physical properties of DNA are broadly exploited by cells to perform the molecular feats necessary for life including storage of information, replication and repair of that information, and regulataion of how that information is expressed. 

In this paper we consider an \textbf{infinite set} of DNA molecules and use tree-hierarchy (introduced in \cite{Rb}) of this set of DNAs to give interactions between neighboring molecules of DNA. 

We study two type models: 

\textit{First model}  is the bubble coalescence in each of interacting DNAs depending on four parameters: (1) temperature; (2) \textit{two} distinct parameter giving the inner interaction of base pairs (in each DNA); (3)  outer interactions of base pairs in a DNA with base pairs of neighboring DNAs. This is the bubble coalescence model in interacting DNAs (shortly: BCI-DNA model). 

\textit{Second model} is the bubble coalescence in a condensed DNA (shortly BCC-DNA) molecules. DNA condensation refers to the process of compacting DNA molecules, which is defined as ``the collapse of extended DNA chains into compact, orderly particles containing only \textbf{one} or a few molecules" (see \cite{TB}).  This model has three parameters, and is the bubble coalescence in one molecule of condensed DNAs, i.e., BCC-DNA model.

For investigation of BCI-DNA and BCC-DNA models we use methods of statistical physics (as in \cite{Robp}, \cite{Rb}, \cite{Rp} and \cite{Rm}), to study its  bubble coalescence thermodynamics.

By tree-hierarchy the DNAs of BCI-DNA and BCC-DNA models are embedded in a Cayley tree. Therefore, their thermodynamics is studied by translation-invariant Gibbs measures (TIGMs) on the Cayley tree.
Note that non-uniqueness of Gibbs measure corresponds to
phase coexistence in the system of DNAs.

The paper is organized as follows.
In Section 2 we give main definitions. 
Section 3 is devoted to BCI-DNA model. In subsection 3.1  we give a system of functional equations, each solution of which
defines a consistent family of finite-dimensional Gibbs distributions and 
guarantees existence of thermodynamic limit for such distributions. This system is very complicated to solve, after some assumptions, in subsection 3.2,  we reduce it to a one-dimensional fixed point problem. Some numerical computations are used to show that the fixed point equation may have up to three solutions. To each such fixed points corresponds a TIGM.  Thus there up to three TIGMs (non-uniqueness - phase transition).

In subsection 3.3  by properties of Markov chains (corresponding to TIGMs) we give the bubble coalescence properties of the model. 
Section 3.4 devoted to biological interpretations of results.

Section 4 is devoted to BCC-DNA model,  we show that for any (admissible) parameters this model has unique TIGM (uniqueness-no-phase transition). 

\section{Preliminaries}

For convenience of a reader let us recall some definitions (see  \cite{R}-\cite{Rm}).\\

{\bf Cayley tree.} The Cayley tree $\Gamma^k$ of order $ k\geq 1 $ is an infinite tree,
i.e., a graph without cycles, such that exactly $k+1$ edges
originate from each vertex. Let $\Gamma^k=(V,L,i)$, where $V$ is the
set of vertices $\Gamma^k$, $L$ the set of edges and $i$ is the
incidence function setting each edge $l\in L$ into correspondence
with its endpoints $x, y \in V$. If $i (l) = \{ x, y \} $, then
the vertices $x$ and $y$ are called the {\it nearest neighbors},
denoted by $l = \langle x, y \rangle $. The distance $d(x,y), x, y
\in V$ on the Cayley tree is the number of edges of the shortest
path from $x$ to $y$:
$$
d (x, y) = \min\{d \,|\, \exists x=x_0, x_1,\dots, x_{d-1},
x_d=y\in V \ \ \mbox {such that} \ \ \langle x_0,
x_1\rangle,\dots, \langle x_{d-1}, x_d\rangle\} .$$

For a fixed $x^0\in V$ we set $ W_n = \ \{x\in V\ \ | \ \ d (x,
x^0) =n \}, $
\begin{equation}\label{p*}
	V_n = \ \{x\in V\ \ | \ \ d (x, x^0) \leq n \},\ \ L_n = \ \{l =
	\langle x, y\rangle \in L \  | \ x, y \in V_n \}.
\end{equation}
For any $x\in V$ denote
$$
W_m(x)=\{y\in V: d(x,y)=m\}, \ \ m\geq 1.
$$

{\bf Group representation of the tree.}
Let $G_k$ be a free product of $k + 1$ cyclic groups of the
second order with generators $a_1, a_2,\dots, a_{k+1}$,
respectively, i.e. $a_i^2=e$, where $e$ is the unit element.

It is known that there exists a one-to-one correspondence between the set of vertices $V$ of the
Cayley tree $\Gamma^k$ and the group $G_k$ (see Chapter 1 of \cite{R} for properties of the group $G_k$).

We consider a normal subgroup $\mathcal H_0\subset G_k$ of infinite index constructed as follows.
Let the mapping $\pi_0:\{a_1,...,a_{k+1}\}\longrightarrow \{e, a_1, a_2\}$ be defined by
$$\pi_0(a_i)=\left\{%
\begin{array}{ll}
	a_i, & \hbox{if} \ \ i=1,2 \\
	e, & \hbox{if} \ \ i\ne 1,2. \\
\end{array}
\right.$$ Denote by $G_1$ the free product of cyclic groups
$\{e,a_1\}, \{e,a_2\}$. Consider $f_0: G_k\to G_1$ defined by
$$f_0(x)=f_0(a_{i_1}a_{i_2}...a_{i_m})=\pi_0(a_{i_1})
\pi_0(a_{i_2})\dots\pi_0(a_{i_m}).$$
Then it is easy to see that $f_0$ is a homomorphism and hence
$\mathcal H_0=\{x\in G_k: \ f_0(x)=e\}$ is a normal subgroup of
infinite index.

Consider the factor group
$$G_k/\mathcal H_0=\{\mathcal H_0, \mathcal H_0(a_1), \mathcal H_0(a_2), \mathcal H_0(a_1a_2), \dots\},$$
where $\mathcal H_0(y)=\{x\in G_k: f_0(x)=y\}$. Denote
$$\mathcal H_n=\mathcal H_0(\underbrace{a_1a_2\dots}_n), \ \ \
\mathcal H_{-n}=\mathcal H_0(\underbrace{a_2a_1\dots}_n).$$
In this notation, the factor group can be represented as
$$ G_k/\mathcal H_0=\{\dots, \mathcal H_{-2}, \mathcal H_{-1}, \mathcal H_0, \mathcal H_1, \mathcal H_2, \dots\}.$$
We introduce the following equivalence relation on the set $G_k$: $x\sim y$ if $xy^{-1}\in \mathcal H_0$.
Then $G_k$ can be partitioned to countably many classes $\mathcal H_i$ of equivalent elements.
The partition of the Cayley tree $\Gamma^2$ w.r.t. $\mathcal H_0$ is shown in
Fig. \ref{fig9} (the elements of the class $\mathcal H_i$, $i\in \mathbb Z$, are merely denoted by $i$).\\

{\bf $\mathbb Z$-path.}
Denote
$$
q_i(x) = |W_1(x)\cap \mathcal H_i|, \ \ x\in G_k,
$$
where $|\cdot|$ is the counting measure of a set.
We note that (see \cite{RI}) if $x\in \mathcal H_m$, then
$$
q_{m-1}(x)=1,  \ \ q_m(x)=k-1, \ \ q_{m+1}(x)=1.
$$
From this fact it follows that
for any $x\in V$, if $x\in \mathcal H_m$ then there is a unique two-side-path (containing $x$) such that
the sequence of numbers of equivalence classes for vertices of this path
in one side are $m, m+1, m+2,\dots$ in the second side the sequence is $m, m-1,m-2,\dots$.
Thus the two-side-path has the sequence of numbers of equivalent classes as $\mathbb Z=\{...,-2,-1,0,1,2,...\}$.
Such a path is called $\mathbb Z$-path (In Fig. \ref{fig9} one can see the unique $\mathbb Z$-paths of each vertex of the tree.)

Since each vertex $x$ has its own $\mathbb Z$-path one can see that
the Cayley tree considered with respect to normal subgroup $\mathcal H_0$ contains infinitely many (countable) set of
$\mathbb Z$-paths.\\
\begin{figure}
	\includegraphics[width=13cm]{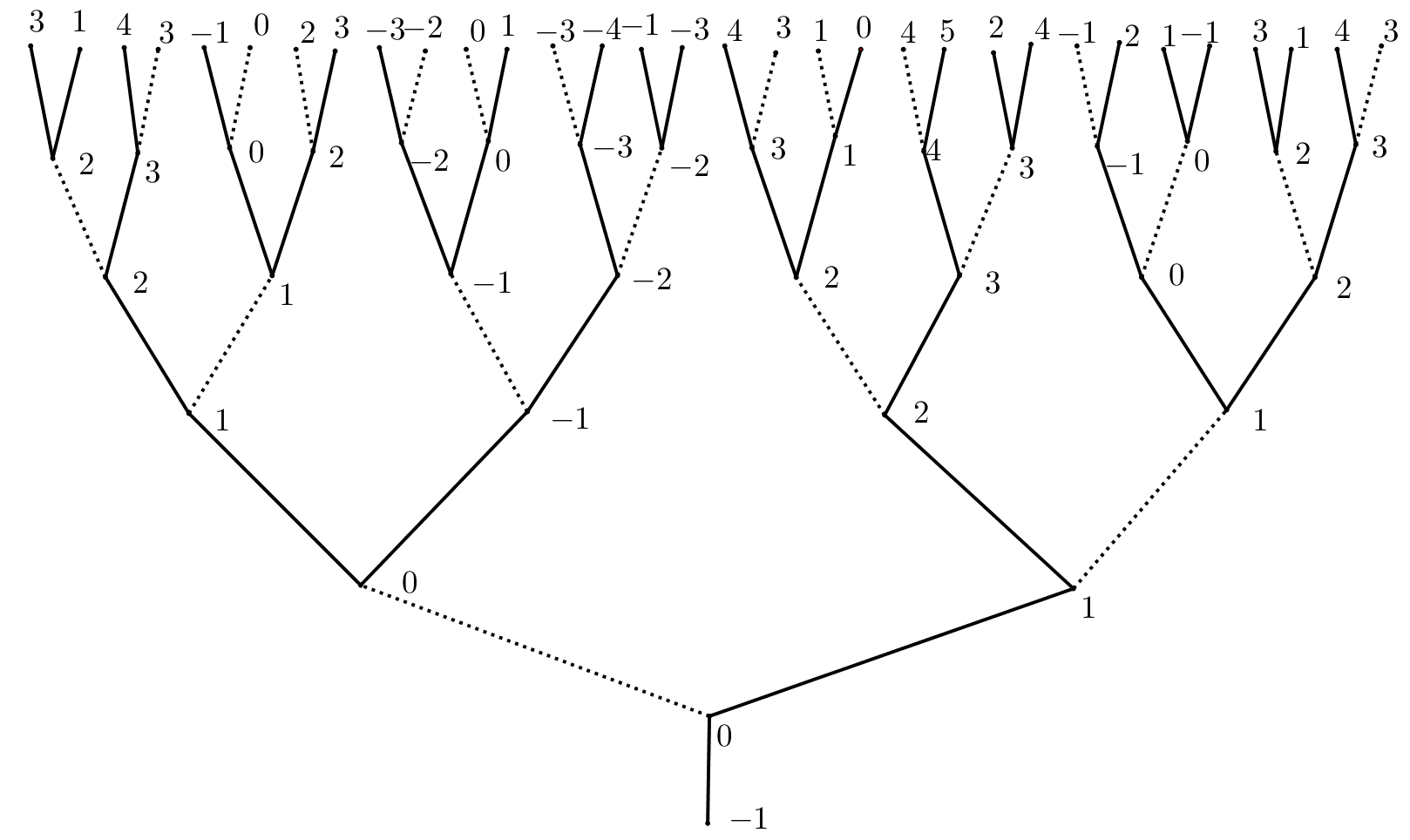}\\
	\caption{The partition of the Cayley tree $\Gamma^2$ w.r.t.
		$\mathcal H_0$, the elements of the class $\mathcal H_i$, $i\in \mathbb Z$,
		are denoted by $i$. $\mathbb Z$-paths are solid lines.}\label{fig9}
\end{figure}

{\bf Configuration space.} Consider spin values from $\Phi=\Psi\times \Phi_q$, with $\Psi=\{b,r\}$ and $\Phi_q=\{1,\dots,q\}$, where $b=blue$, $r=red$, $q\geq 1$. 

A configuration is any mapping $s: x\in V:\to s(x)=(\varphi(x), \sigma(x))\in \Phi$. Denote by $\Omega=\Phi^V$ the set of configurations. 

Configurations in
$V_n$  are defined analogously and the set of all
configurations in $V_n$ is denoted by $\Omega_n$.

The restriction of a configuration on
a $\mathbb Z$-path is called {\it a DNA}. 
Since there are countably many $\mathbb Z$-paths we have a countably many distinct DNAs.

In Fig.\ref{fta} we give a collection of interacting DNAs: on each $\mathbb Z$-path red points correspond  to soft zones (see Fig. \ref{bub}) and blue points are presenting the barrier zones.\\ 

\begin{figure}
	\includegraphics[width=15cm]{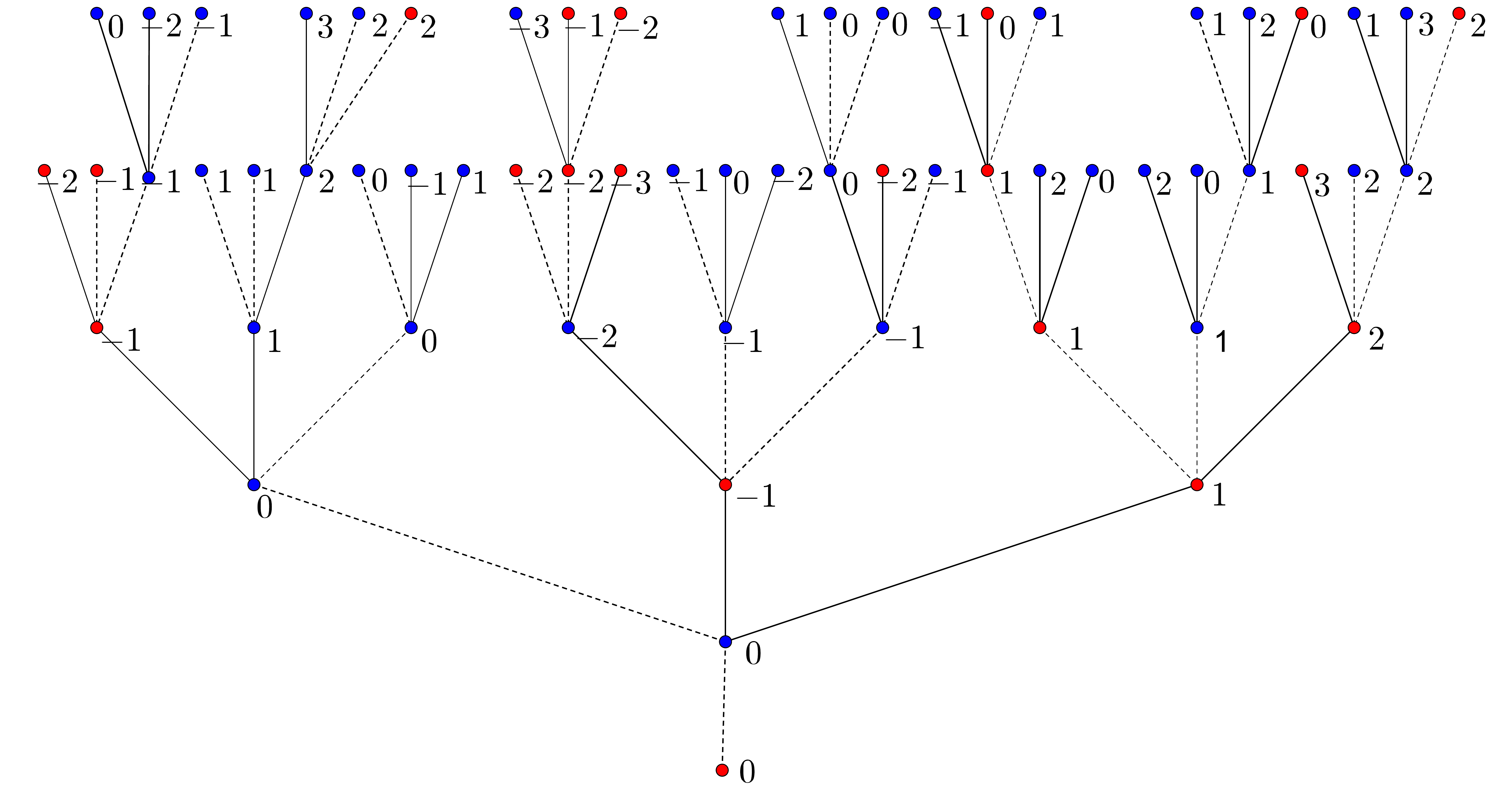}\\
	\caption{A part of Cayley tree of order four. Dashed edges do not belong to a DNA. Each solid edge belongs to a $\mathbb Z$-path, i.e., is a region of a DNA. Neighboring two DNAs interact if the color of neighboring vertices (separating DNAs) are the same. But neighboring vertices on a DNA interact if they have distinct colors.}\label{fta}
\end{figure}

{\bf Tree-hierarchy of the set of DNAs.} 

Define a Cayley tree
hierarchy of the set of DNAs as follows.

We say that two DNA {\it are neighbors} if there is an edge (of the Cayley tree) such that
one its endpoint belongs to the first DNA and
another endpoint of the edge belongs to the second DNA. By our
construction it is clear (see Fig. \ref{fig9}) that
such an edge is unique for each neighboring pair of DNAs. This edge
has equivalent endpoints, i.e. both endpoints belong to the same class
$\mathcal H_m$ for some $m\in \mathbb Z$.

Moreover these countably infinite set of DNAs have a hierarchy that:

(i) any two DNA do not intersect.

(ii) each DNA has its own countably many set of neighboring DNAs;

(iii) for any two neighboring DNAs, say $D_1$ and $D_2$, there exists a unique
edge $l=l(D_1,D_2)=\langle x,y\rangle$ with $x\sim y$ which connects DNAs;

(iv)  for any finite $n\geq 1$ the ball $V_n$ has intersection only with finitely many DNAs.\\

{\bf The Hamiltonian of BCI-DNA model.} 

We consider the following model of the energy of the configuration $s\in \Omega$ of a set of DNAs:
\begin{equation}\label{h}
	H(s)=\sum_{\langle x,y\rangle\in L}f_{x,y}(s(x),s(y)),
\end{equation}
where $s(x)=(\varphi(x), \sigma(x))$ and 

\begin{equation}\label{hd}f_{x,y}(s(x), s(y))=\left\{\begin{array}{lll}
\left(1-\delta_{\varphi(x)\varphi(y)}\right)\left(J_b\delta_{b\varphi(x)}
B(\sigma(x),\sigma(y))+J_r\delta_{r\varphi(x)}
R(\sigma(x),\sigma(y))\right),\\[2mm]
\ \ \ \ \ \ \ \ \ \ \ \ \ \ \ \ \ \ \ \ \ \ \ \ \ \ \ \mbox{if} \ \ \langle x,y\rangle\in \mathbb Z-path\\[3mm]
J\delta_{\varphi(x)\varphi(y)}\delta_{\sigma(x),\sigma(y)}, \ \ \mbox{if} \ \ \langle x,y\rangle\notin \mathbb Z-path
	\end{array}
	\right.
\end{equation}
$J>0$ is a coupling constant between neighboring DNAs, $\delta$ is the Kronecker delta, $J_b<0$, $J_r>0$ are parameters, and $B, R: \Phi_q\times\Phi_q \to \mathbb R$ are non-negative functions, which give interaction between DNA base pairs.

\begin{rk} We note that
	\begin{itemize}
\item Hamiltonian (\ref{h}) consists interactions between base pairs of a DNA if the base pairs are in distinct zone (see Fig. \ref{bub}), i.e.,  interactions in a DNA exist between red and blue points. But interaction between neighboring DNAs is given by connecting them edge (a dashed edge in Fig.\ref{fta}) when the endpoints of this edge have the same color. 
\item In this paper, for simplicity, we mainly consider the case when functions $B$ and $R$ are given by the SOS (gradient) function (i.e., SOS model, see \cite{KRsos}) and by Kronecker delta (i.e., Potts model, see \cite{KRK}, \cite{KR}). In case $q=2$ the BCI-DNA model combines Potts models defined on DNA's edges and dashed edges of the Cayley tree (see \cite[Section 2.4]{Robp} for the relevance of the Potts models in several applied fields).  
\end{itemize} 
\end{rk}

{\bf The Hamiltonian of BCC-DNA model.} 

In this model \textbf{any} path of the Cayley tree is considered as a part of DNA, the full Cayley tree is considered  as one molecule of a condensed DNA. 

We consider the following BCC-DNA model of the energy of the configuration $s\in \Omega$:
\begin{equation}\label{hc}
	H(s)=\sum_{\langle x,y\rangle\in L}g(s(x),s(y)),
\end{equation}
where $s(x)=(\varphi(x), \sigma(x))$ and 
\begin{equation}\label{hdc}g(s(x), s(y))=
		\left(1-\delta_{\varphi(x)\varphi(y)}\right)\left(J_b\delta_{b\varphi(x)}
		B(\sigma(x),\sigma(y))+J_r\delta_{r\varphi(x)}
		R(\sigma(x),\sigma(y))\right),		
	\end{equation}
where $J_b<0$, $J_r>0$ are parameters, and $B, R: \Phi_q\times \Phi_q\to \mathbb R$ are non-negative functions, which give interaction between DNA base pairs.

\section{Thermodynamics of the BCI-DNA model}

\subsection{System of functional equations of finite dimensional distributions}
Let $\Omega_n$ be the set of all
configurations on $V_n$.

Define a finite-dimensional distribution of a probability measure $\mu$ on $\Omega_n$ as
\begin{equation}\label{*}
	\mu_n(s_n)=Z_n^{-1}\exp\left\{\beta H_n(s_n)+\sum_{y\in W_n}h_{\varphi(y),\sigma(y), y}\right\},
\end{equation}
where $s_n(x)=(\varphi(x),\sigma(x))$, $x\in V_n$, $\beta=1/T$, $T>0$ is temperature,  $Z_n^{-1}$ is the normalizing factor,
$$\{h_{a, i, x}\in \mathbb R, a\in \Psi, i\in\Phi_q, \, x\in V\}$$ is a collection of real numbers and
$$H_n(s_n)=\sum_{\langle x,y\rangle\in L_n}f_{x,y}(s_n(x),s_n(y)).$$

We say that the probability distributions (\ref{*}) are compatible if for all
$n\geq 1$ and $s_{n-1}\in \Omega_{n-1}$:
\begin{equation}\label{**}
	\sum_{\omega_n\in \Omega_{W_n}}\mu_n(s_{n-1}\vee \omega_n)=\mu_{n-1}(s_{n-1}).
\end{equation}
Here $s_{n-1}\vee \omega_n$ is the concatenation of the configurations.

For $x\in W_{n-1}$ denote
$$S(x)=\{t\in W_n: \langle x,t\rangle\}.$$

For $ x\in V$ we denote by $x_{\downarrow}$ the unique point of the set $\{y\in V:\langle x,y\rangle\}\setminus S(x)$.

It is easy to see that
$$S(x)\cap \mathbb Z-{\rm path}=\left\{\begin{array}{lll}
	\{x_0, x_1\}\subset V, \ \ \mbox{if} \ \  \langle x_\downarrow, x\rangle\notin \mathbb Z-{\rm path}\\[2mm]
	\{x_1\}\subset V, \ \ \ \ \ \ \mbox{if} \ \  \langle x_\downarrow, x\rangle\in \mathbb Z-{\rm path}
\end{array}
\right..$$
We denote
$$S_0(x)=S(x)\setminus \{x_0, x_1\}, \ \ \langle x_\downarrow, x\rangle\notin  \mathbb Z-{\rm path},$$
$$S_1(x)=S(x)\setminus \{x_1\}, \ \ \langle x_\downarrow, x\rangle\in  \mathbb Z-{\rm path}.$$

The following theorem gives a criterion for compatibility of finite-dimensional distributions.

\begin{thm}\label{ei} Probability distributions
	$\mu_n(s_n)$, $n=1,2,\ldots$, in
	(\ref{*}) are compatible iff for any $x\in V\setminus {x^0}$
	the following equations hold:
	
	if $\langle x_\downarrow, x\rangle\notin \mathbb Z-{\rm path}$ then
$$
	z_{b,i,x}=\prod_{t\in \{x_0, x_1\}}{1+\eta^{R(i,q)}\hat z_{r,q,t}+\sum_{j=1}^{q-1}\left(\hat z_{b,j,t}+\eta^{R(i,j)} \hat z_{r,j,t}\right)\over 
		1+\eta^{R(q,q)}\hat z_{r,q,t}+\sum_{j=1}^{q-1}\left(\hat z_{b,j,t}+\eta^{R(q,j)}\hat z_{r,j,t}\right)}\times $$
	\begin{equation}\label{B}
	\prod_{y\in S_0(x)}{1+z_{r,q,y}+\sum_{j=1}^{q-1}\left(\theta^{\delta_{ij}}z_{b,j,y}+z_{r,j,y}\right)\over 
		\theta+z_{r,q,y}+\sum_{j=1}^{q-1}\left(z_{b,j,y}+z_{r,j,y}\right)}, \ \ i\in \Phi_{q-1};
\end{equation}
$$	z_{r,i,x}=\prod_{t\in \{x_0, x_1\}}{\zeta^{B(i,q)}+\hat z_{r,q,t}+\sum_{j=1}^{q-1}
	\left(\zeta^{B(i,j)}\hat z_{b,j,t}+\hat z_{r,j,t}\right)\over 
		1+\eta^{R(q,q)}\hat z_{r,q,t}+\sum_{j=1}^{q-1}\left(\hat z_{b,j,t}+\eta^{R(q,j)}\hat z_{r,j,t}\right)}$$
		\begin{equation}\label{R}
			\prod_{y\in S_0(x)}{1+z_{r,q,y}+\sum_{j=1}^{q-1}\left(z_{b,j,y}+\theta^{\delta_{ij}}z_{r,j,y}\right)\over 
			\theta+z_{r,q,y}+\sum_{j=1}^{q-1}\left(z_{b,j,y}+z_{r,j,y}\right)},  \ \ i\in \Phi_q.
\end{equation}
	if $\langle x_\downarrow, x\rangle\in \mathbb Z-{\rm path}$ then
$$
\hat z_{b,i,x}={1+\eta^{R(i,q)}\hat z_{r,q,x_1}+\sum_{j=1}^{q-1}\left(\hat z_{b,j,x_1}+\eta^{R(i,j)}\hat z_{r,j,x_1}\right)\over 
	1+\eta^{R(q,q)}\hat z_{r,q,x_1}+\sum_{j=1}^{q-1}\left(\hat z_{b,j,x_1}+\eta^{R(q,j)}\hat z_{r,j,x_1}\right)}\times $$
\begin{equation}\label{Bz}
	\prod_{y\in S_1(x)}{1+z_{r,q,y}+\sum_{j=1}^{q-1}\left(\theta^{\delta_{ij}}z_{b,j,y}+z_{r,j,y}\right)\over 
		\theta+z_{r,q,y}+\sum_{j=1}^{q-1}\left(z_{b,j,y}+z_{r,j,y}\right)}, \ \ i\in \Phi_{q-1};
\end{equation}
$$	\hat z_{r,i,x}={\zeta^{B(i,q)}+\hat z_{r,q,x_1}+\sum_{j=1}^{q-1}
	\left(\zeta^{B(i,j)}\hat z_{b,j,x_1}+\hat z_{r,j,x_1}\right)\over 
	1+\eta^{R(q,q)}\hat z_{r,q,x_1}+\sum_{j=1}^{q-1}\left(\hat z_{b,j,x_1}+\eta^{R(q,j)}\hat z_{r,j,x_1}\right)}$$
\begin{equation}\label{Rz}
	\prod_{y\in S_1(x)}{1+z_{r,q,y}+\sum_{j=1}^{q-1}\left(z_{b,j,y}+\theta^{\delta_{ij}}z_{r,j,y}\right)\over 
		\theta+z_{r,q,y}+\sum_{j=1}^{q-1}\left(z_{b,j,y}+z_{r,j,y}\right)},  \ \ i\in \Phi_{q}.
\end{equation}
where $\theta=e^{J\beta}$, $\zeta=e^{J_b\beta}$ and $\eta=e^{J_r\beta}$. 
\begin{equation}\label{alp}\begin{array}{ll}
	z_{a,i,x}=\exp\left(h_{a, i, x}-h_{b, q, x}\right), a\in \Psi, i\in \Phi_q, \, \langle x_\downarrow, x\rangle\notin \mathbb Z-{\rm path}; \\[2mm]
		\hat z_{a,i,x}=\exp\left(h_{a, i, x}-h_{b, q, x}\right), a\in \Psi, i\in \Phi_q, \, \langle x_\downarrow, x\rangle\in \mathbb Z-{\rm path}.
	\end{array}
\end{equation}
\end{thm}
\begin{proof} Similar to the proof of
	Theorem 5.1 in \cite{R}, see also the proof of Theorem 1 in \cite{RR}.
\end{proof}

This is $4q-2$ dimensional non-liner system of functional equation. The unknown functions are defined on vertices of the tree and take strictly positive real values.

\subsection{Constant unknown functions} 
 
 In general, it is very difficult to find
solutions of the system (\ref{B}), (\ref{R}), (\ref{Bz}), (\ref{Rz}) . Therefore one can solve it in class of translation invariant (constant) functions. That is we assume that our unknown functions do not depend on the vertices of tree:
$$z_{b,i,x}\equiv u_i, \ \ \mbox{for all} \ \ x\in V, \ \ i\in \Phi_{q-1};$$ 
$$z_{r,i,x}\equiv v_i, \ \ \mbox{for all} \ \ x\in V, \ \ i\in \Phi_{q};$$ 
$$\hat z_{b,i,x}\equiv \hat u_i, \ \ \mbox{for all} \ \ x\in V, \ \ i\in \Phi_{q-1};$$ 
$$\hat z_{r,i,x}\equiv \hat v_i, \ \ \mbox{for all} \ \ x\in V, \ \ i\in \Phi_{q}.$$ 
Then system (\ref{B}), (\ref{R}), (\ref{Bz}), (\ref{Rz}) is reduced to 
\begin{equation}\label{Bt}
	u_i=\left({1+\eta^{R(i,q)}\hat v_q+\sum_{j=1}^{q-1}\left(\hat u_j+\eta^{R(i,j)}\hat v_j\right)\over 
		1+\eta^{R(q,q)}\hat v_q+\sum_{j=1}^{q-1}\left(\hat u_j+
		\eta^{R(q,j)}\hat v_j\right)}\right)^2\left({1+v_q+\sum_{j=1}^{q-1}\left(\theta^{\delta_{ij}}u_j+v_j\right)\over 
		\theta+v_q+\sum_{j=1}^{q-1}\left(u_j+v_j\right)}\right)^{k-2}, \ i\in \Phi_{q-1};
\end{equation}
\begin{equation}\label{Rt}
	v_i=\left({\zeta^{B(i,q)}+\hat v_q+\sum_{j=1}^{q-1}
		\left(\zeta^{B(i,j)}\hat u_j+\hat v_j\right)\over 
		1+\eta^{R(q,q)}\hat v_q+\sum_{j=1}^{q-1}\left(\hat u_j+
	\eta^{R(q,j)}\hat v_j\right)}\right)^2\left({1+v_q+\sum_{j=1}^{q-1}\left(u_j+\theta^{\delta_{ij}}v_j\right)\over 
	\theta+v_q+\sum_{j=1}^{q-1}\left(u_j+v_j\right)}\right)^{k-2}, \  i\in \Phi_{q}.
\end{equation}
\begin{equation}\label{Btz}
	\hat u_i=\left({1+\eta^{R(i,q)}\hat v_q+\sum_{j=1}^{q-1}\left(\hat u_j+\eta^{R(i,j)}\hat v_j\right)\over 
		1+\eta^{R(q,q)}\hat v_q+\sum_{j=1}^{q-1}\left(\hat u_j+
		\eta^{R(q,j)}\hat v_j\right)}\right)\left({1+v_q+\sum_{j=1}^{q-1}\left(\theta^{\delta_{ij}}u_j+v_j\right)\over 
		\theta+v_q+\sum_{j=1}^{q-1}\left(u_j+v_j\right)}\right)^{k-1}, \ i\in \Phi_{q-1};
\end{equation}
\begin{equation}\label{Rtz}
	\hat v_i=\left({\zeta^{B(i,q)}+\hat v_q+\sum_{j=1}^{q-1}
		\left(\zeta^{B(i,j)}\hat u_j+\hat v_j\right)\over 
		1+\eta^{R(q,q)}\hat v_q+\sum_{j=1}^{q-1}\left(\hat u_j+
		\eta^{R(q,j)}\hat v_j\right)}\right)\left({1+v_q+\sum_{j=1}^{q-1}\left(u_j+\theta^{\delta_{ij}}v_j\right)\over 
		\theta+v_q+\sum_{j=1}^{q-1}\left(u_j+v_j\right)}\right)^{k-1}, \  i\in \Phi_{q}.
\end{equation}
Now one can choose concrete functions $B$ and $R$ and then try to solve corresponding system of equations (\ref{Bt}), (\ref{Rt}), (\ref{Btz}), (\ref{Rtz}). 

For the Bubble coalescence model it seems reasonable to take 
these functions as
\begin{equation}\label{fc} 
B(i,j)=|i-j|, \ \ R(i,j)=1-\delta_{ij}.
\end{equation}
Then the system is simplified as
 \begin{equation}\label{Btq}
 	u_i=\left({1+\eta \hat v_q+\sum_{j=1}^{q-1}\left(\hat u_j+\eta^{1-\delta_{ij}}\hat v_j\right)\over 
 		1+\hat v_q+\sum_{j=1}^{q-1}\left(\hat u_j+
 		\eta \hat v_j\right)}\right)^2\left({1+v_q+\sum_{j=1}^{q-1}\left(\theta^{\delta_{ij}}u_j+v_j\right)\over 
 		\theta+v_q+\sum_{j=1}^{q-1}\left(u_j+v_j\right)}\right)^{k-2}, \ i\in \Phi_{q-1};
 \end{equation}
 \begin{equation}\label{Rtq}
 	v_i=\left({\zeta^{q-i}+\hat v_q+\sum_{j=1}^{q-1}
 		\left(\zeta^{|i-j|}\hat u_j+\hat v_j\right)\over 
 		1+\hat v_q+\sum_{j=1}^{q-1}\left(\hat u_j+
 		\eta \hat v_j\right)}\right)^2\left({1+v_q+\sum_{j=1}^{q-1}\left(u_j+\theta^{\delta_{ij}}v_j\right)\over 
 		\theta+v_q+\sum_{j=1}^{q-1}\left(u_j+v_j\right)}\right)^{k-2}, \ \ i\in \Phi_{q}.
 \end{equation}
 \begin{equation}\label{Btqz}
	\hat u_i=\left({1+\eta \hat v_q+\sum_{j=1}^{q-1}\left(\hat u_j+\eta^{1-\delta_{ij}}\hat v_j\right)\over 
		1+\hat v_q+\sum_{j=1}^{q-1}\left(\hat u_j+
		\eta \hat v_j\right)}\right)\left({1+v_q+\sum_{j=1}^{q-1}\left(\theta^{\delta_{ij}}u_j+v_j\right)\over 
		\theta+v_q+\sum_{j=1}^{q-1}\left(u_j+v_j\right)}\right)^{k-1}, \ i\in \Phi_{q-1};
\end{equation}
\begin{equation}\label{Rtqz}
	\hat v_i=\left({\zeta^{q-i}+\hat v_q+\sum_{j=1}^{q-1}
		\left(\zeta^{|i-j|}\hat u_j+\hat v_j\right)\over 
		1+\hat v_q+\sum_{j=1}^{q-1}\left(\hat u_j+
		\eta \hat v_j\right)}\right)\left({1+v_q+\sum_{j=1}^{q-1}\left(u_j+\theta^{\delta_{ij}}v_j\right)\over 
		\theta+v_q+\sum_{j=1}^{q-1}\left(u_j+v_j\right)}\right)^{k-1}, \ \ i\in \Phi_{q}.
\end{equation}
For simplicity we consider the case $q=2$, meaning, for example,  that 
$$1=A+T \ \ \mbox{and} \ \ 2= C+G.$$ 
Then from system (\ref{Btq}), (\ref{Rtq}), (\ref{Btqz}), (\ref{Rtqz}) we get ($u=u_1$, $v=v_1$, $w=v_2$, $\hat u=\hat u_1$, $\hat v=\hat v_1$, $\hat w=\hat v_2$):
\begin{equation}\label{BR2}
	\begin{array}{lll}
	u=({1+\hat u+\hat v+\eta \hat w\over 1+\hat u+\eta \hat v+\hat w})^2({1+\theta u+v+w\over \theta+u+v+w})^{k-2}\\[2mm]
	v=({\zeta+\hat u+\hat v+\hat w\over 1+\hat u+\eta \hat v+\hat w})^2({1+ u+\theta v+w\over \theta+u+v+w})^{k-2}\\[2mm]
	w=({1+\zeta \hat u+\hat v+\hat w\over 1+\hat u+\eta \hat v+\hat w})^2({1+ u+v+\theta w\over \theta+u+v+w})^{k-2},
	\end{array}
\end{equation}
\begin{equation}\label{BR2z}
	\begin{array}{lll}
		\hat u=({1+\hat u+\hat v+\eta \hat w\over 1+\hat u+\eta \hat v+\hat w})({1+\theta u+v+w\over \theta+u+v+w})^{k-1}\\[2mm]
		\hat v=({\zeta+\hat u+\hat v+\hat w\over 1+\hat u+\eta \hat v+\hat w})({1+ u+\theta v+w\over \theta+u+v+w})^{k-1}\\[2mm]
		\hat w=({1+\zeta \hat u+\hat v+\hat w\over 1+\hat u+\eta \hat v+\hat w})({1+ u+v+\theta w\over \theta+u+v+w})^{k-1},
	\end{array}
\end{equation}
where $\theta>1$, $0<\zeta<1$, $\eta>1$ and $u,v,w>0$.

To solve this system we note that this is fixed point equation for the operator $F: \mathbb R^6_+\to \mathbb R^6_+$ defined by 
\begin{equation}\label{BRo}
	F:\ 	\begin{array}{llllll}
			u'=({1+\hat u+\hat v+\eta \hat w\over 1+\hat u+\eta \hat v+\hat w})^2({1+\theta u+v+w\over \theta+u+v+w})^{k-2}\\[2mm]
			v'=({\zeta+\hat u+\hat v+\hat w\over 1+\hat u+\eta \hat v+\hat w})^2({1+ u+\theta v+w\over \theta+u+v+w})^{k-2}\\[2mm]
			w'=({1+\zeta \hat u+\hat v+\hat w\over 1+\hat u+\eta \hat v+\hat w})^2({1+ u+v+\theta w\over \theta+u+v+w})^{k-2},\\[2mm]
			\hat u'=({1+\hat u+\hat v+\eta \hat w\over 1+\hat u+\eta \hat v+\hat w})({1+\theta u+v+w\over \theta+u+v+w})^{k-1}\\[2mm]
		\hat v'=({\zeta+\hat u+\hat v+\hat w\over 1+\hat u+\eta \hat v+\hat w})({1+ u+\theta v+w\over \theta+u+v+w})^{k-1}\\[2mm]
		\hat w'=({1+\zeta \hat u+\hat v+\hat w\over 1+\hat u+\eta \hat v+\hat w})({1+ u+v+\theta w\over \theta+u+v+w})^{k-1}.
	\end{array}
\end{equation}
Define 
$$M=\{(u, v, w, \hat u, \hat v, \hat w)\in \mathbb R^6_+: u=1, v=w, \hat u=1, \hat v=\hat w\}.$$
\begin{lemma}
The set $M$ is invariant with respect to $F$.
\end{lemma}
\begin{proof}
	It is straightforward to see that if $(u, v, w, \hat u, \hat v, \hat w)\in M$ then $(u', v', w', \hat u', \hat v', \hat w')\in M$, i.e., $F(M)\subset M$.  
\end{proof}

Let us reduce operator $F$ on the invariant set $M$, then the fixed points on $M$ are given by the solutions of the following system 
\begin{equation}\label{MR}
	 	\begin{array}{ll}
		v=({\zeta+1+2\hat v\over 2+(1+\eta) \hat v})^2({2+ (1+\theta) v\over \theta+1+2v})^{k-2}\\[2mm]
\hat v=({\zeta+1+2\hat v\over 2+(1+\eta) \hat v})({2+ (1+\theta) v\over \theta+1+2v})^{k-1}.
\end{array}\end{equation}
From the second equation of this system we get
$${\zeta+1+2\hat v\over 2+(1+\eta) \hat v}=\hat v\, ({2+ (1+\theta) v\over \theta+1+2v})^{1-k}.$$
Substituiting this in the first equation of (\ref{MR}) we obtain 
\begin{equation}\label{hat}
	\hat v =\sqrt{v}\, ({2+ (1+\theta) v\over \theta+1+2v})^{k/2}.
\end{equation}
Consequently, from the first equation of (\ref{MR}) we get 
\begin{equation}\label{ff}
	v=f(v):=f(v,\theta, \eta,\zeta,k),
\end{equation}	
	where 
	$$f(v,\theta, \eta,\zeta,k):=
\left({(\zeta+1)(\theta+1+2v)^{k/2}+2\sqrt{v} (2+ (1+\theta) v)^{k/2}\over 2(\theta+1+2v)^{k/2}+(1+\eta) \sqrt{v} (2+ (1+\theta) v)^{k/2}}\right)^2\left({2+ (1+\theta) v\over \theta+1+2v}\right)^{k-2}.$$
This is very complicated equation depending on four parameters:  $$k\geq 2, \ \ \theta>1, \ \ \zeta\in (0,1), \ \ \eta>1.$$
But our reduction the system to equation (\ref{ff}) with one unknown is very useful to solve the system numerically: one can take concrete values of parameters and then a computer gives all corresponding solutions.  

We are interested to the values of parameters when the equation (\ref{ff}) has more than one solutions. Since the problem is very difficult, we choose concrete values of parameters as 
\begin{equation}\label{tze}
	\theta=3, \ \ \zeta=0.5, \ \ \eta=1.05
	\end{equation}
and change values of 
$k=2,3,4,5,6,7.$ The  table (see Fig.\ref{ft}) shows solutions to (\ref{MR}) corresponding to numerical solution of (\ref{ff}) and putting it in (\ref{hat}).
\begin{rk} 	
	\begin{itemize}
	\item For fixed $\beta={1\over T}$ (i.e., fixed temperature) condition (\ref{tze}) is a condition on parameters of the model:
		$$J=T\log(3)>0, \ \ J_b=T\log(0.5)<0, \ \ J_r=T\log(1.05)>0.$$
	
	\item Numerical analysis of system (\ref{MR}) for the case $k\geq 8$ also shows that for the fixed parameters $\theta=3, \ \zeta=0.5, \ \eta=1.05$, there are exactly 3 positive solutions. Minimal solution goes to zero, maximal solution goes to infinity when $k\to \infty$, for example, when $k=20$ we have the following exactly three values of $v$:
	$$0.0000021459,\ \ 1.054575702, \ \ 249485.2116$$
	 Since non-uniqueness appears when $k\geq 6$, there is no any hope to show this analytically.  But numerical results  which we have for $k=2,\dots,7$  are  already
interesting enough to see biological interpretations of our results. 
\end{itemize}
\end{rk}
\begin{figure}
	\includegraphics[width=13cm]{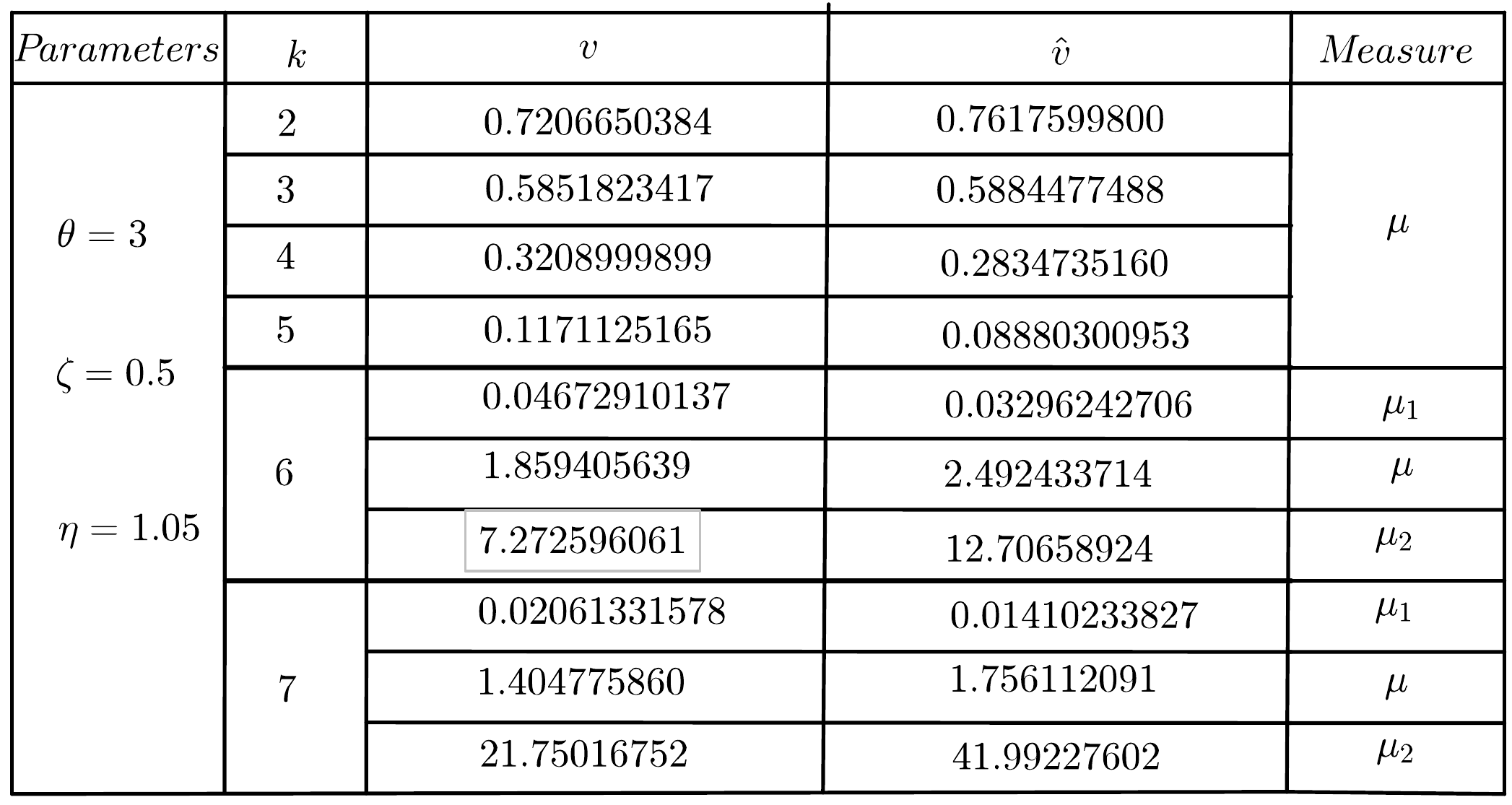}\\
	\caption{The approximate values of solutions $v$ and $\hat v$ corresponding to concrete values of parameters as 
		$\theta=3, \ \ \zeta=0.5, \ \ \eta=1.05$ and for $k=2,3,4,5,6,7.$ In case of $k=2,3,4,5$ there is unique solution. But for $k=6$ and $k=7$ there are exactly three solutions. The corresponding measures are denoted as in the last column.}\label{ft}
\end{figure}

 \subsection{Markov chains corresponding to solutions given in the  table}

We note that the solutions \begin{equation}\label{zab}\begin{array}{ll}
		z_{a,i,x}, a\in \Psi, i\in \Phi_q, \, \langle x_\downarrow, x\rangle\notin \mathbb Z-{\rm path}; \\[2mm]
		\hat z_{a,i,x}, a\in \Psi, i\in \Phi_q, \, \langle x_\downarrow, x\rangle\in \mathbb Z-{\rm path}.
	\end{array}
\end{equation} define a boundary
law (\cite{BR}, \cite[Chapter 12]{Ge}) of the biological system of DNAs.

For marginals on an edge $l=\langle x,y\rangle$,  
considering a boundary law  we have in the case of Hamiltonian (\ref{h}) that
$$\mu[s(x)=(a,i), s(y)=(c,j)]= \frac{1}{Z} z_{a,i} \exp(f_{xy}(s(x),s(y)) z_{c,j},$$
where $Z$ is normalizing factor and $z_{a,x}$ takes values depending on
the relation of $l$ to ${\mathbb Z}- {\rm path}$.

From this,  the relation between the boundary law and the transition matrix
for the associated tree-indexed Markov chain (Gibbs measure) is  obtained from the formula of the conditional probability.

Now we are interested to study thermodynamics of 
bubble  coalescence corresponding to solutions given in the table. For this reason we study Markov chains on the sub-tree
consisting edges which are not a $\mathbb Z$-path and separately Markov chains on the $\mathbb Z$-paths. 

Here we define these two Markov chains.
For $q=2$, for simplify of notations,  let us denote the single-site values of the 
configuration as 
\begin{equation}\label{sn}
1:=(b,1), \ \ 2:=(b,2), \ \ 3:=(r,1), \ \ 4:=(r,2).
\end{equation}
\begin{itemize}
	\item Since our solutions do not depend on vertices we define tree-indexed homogeneous Markov chain with states
	$\{1,2,3,4\}$ (defined in (\ref{sn})) with transition matrix $\mathbf P=\left(P_{ij}\right)$,
	where $P_{ij}$ is the probability to go from a state $i$ at a vertex to a state $j$ at the neighboring vertex of the tree.
	Using solutions $(u, v, w, \hat u, \hat v, \hat w)$  we write the matrices (if the edge does not belong to ${\mathbb Z}-{\rm path}$):
	$$\mathbf P=\left(\begin{array}{cccc}
		{\theta u\over Z_1}& {1\over Z_1}& {v\over Z_1}& {w\over Z_1}\\[2mm]
			{u\over Z_2}& {\theta\over Z_2}& {v\over Z_2}& {w \over Z_2}\\[2mm]
		{u\over Z_3}& {1\over Z_3}& {\theta v\over Z_3}& {w\over Z_3}\\[2mm]
		{u\over Z_4}& {1\over Z_4}& {v\over Z_4}& {\theta w\over Z_4}
		\end{array}\right).
	$$
	Here $Z_1=1+\theta u+v+w$,	$Z_2=\theta+u+v+w$, $Z_3=1+ u+\theta v+w$, $Z_4=1+u+ v+\theta w$,.
	
	\item Define tree-indexed Markov chain on $\mathbb Z$-path:
	$$\mathbf Q=\left(\begin{array}{cccc}
		{\hat u\over Y_1}& {1\over Y_1}& {\hat v\over Y_1}& {\eta \hat w\over Y_1}\\[2mm]
			{\hat u\over Y_2}& {1\over Y_2}& {\eta \hat v\over Y_2}& {\hat w\over Y_2}\\[2mm]
		{\hat u\over Y_3}& {\zeta\over Y_3}& {\hat v\over Y_3}& {\hat w\over Y_3}\\[2mm]
		{\zeta \hat u\over Y_4}& {1\over Y_4}& {\hat v\over Y_4}& {\hat w\over Y_4}
	\end{array}\right).
	$$
	Here $Y_1=1+\hat u+\hat v+\eta \hat w$, $Y_2=1+\hat u+\eta\hat v+\hat w$, $Y_3=\zeta+\hat u+\hat v+\hat w$,
	$Y_4=1+\zeta\hat u+\hat v+\hat w$.\\
	
Compute matrices 	$\mathbf P$ and $\mathbf Q$ for 
the case $k=6$ and 	$\theta=3, \ \ \zeta=0.5, \ \ \eta=1.05$. Since for these values of parameters we have three solutions, denote by 	$\mathbf P_i$ and $\mathbf Q_i$, $i=1,2,3$ the corresponding matrices. Denote by $p_i$ (resp. $q_i$) the stationary probability vector of $\mathbf P_i$ (resp. $\mathbf Q_i$). 

In the case of non-uniqueness of Gibbs measure
(and corresponding Markov chains) we have different stationary states for different measures. 
Using the values given in the Table we get:

{\bf Case of measure $\mu_1$:}
$$\mathbf P_1= \left(\begin{array}{cccc}
	0.733& 0.244& 0.0115&0.0115\\[2mm]
		0.244&0.733&0.0115&0.0115\\[2mm]
	0.457& 0.457& 0.064& 0.022\\[2mm]
	0.457& 0.457& 0.022& 0.064
		\end{array}\right), \ \ \mathbf Q_1= \left(\begin{array}{cccc}
		0.484& 0.484& 0.015&0.017\\[2mm]
			0.484&0.484&0.017&0.015\\[2mm]
		0.638& 0.319& 0.021& 0.021\\[2mm]
		0.319& 0.638& 0.021& 0.021
	\end{array}\right),
	$$
$$p_1=(0.488, 0.488, 0.012, 0.012), \ \ q_1=(0.484, 0.484, 0.016, 0.016).
$$
{\bf Case of measure $\mu$:}
$$\mathbf P_2= \left(\begin{array}{cccc}
	0.389& 0.129& 0.241&0.241\\[2mm]
		0.129&0.389&0.241&0.241\\[2mm]
	0.106& 0.106& 0.591& 0.197\\[2mm]
	0.106& 0.106& 0.197& 0.591
\end{array}\right), \ \ \mathbf Q_2= \left(\begin{array}{cccc}
	0.141& 0.141& 0.350&0.368\\[2mm]
		0.141&0.141&0.368&0.350\\[2mm]
	0.154& 0.078& 0.384& 0.384\\[2mm]
	0.078& 0.154& 0.384& 0.384
\end{array}\right),
$$
$$p_2=(0.153, 0.153, 0.347, 0.347), \ \ q_2=(0.122, 0.122, 0.378, 0.378).
$$
{\bf Case of measure $\mu_2$:}
$$\mathbf P_3= \left(\begin{array}{cccc}
	0.162& 0.054& 0.392&0.392\\[2mm]
		0.054&0.162&0.392&0.392\\[2mm]
	0.032& 0.032& 0.702& 0.234\\[2mm]
	0.032& 0.032& 0.234& 0.702
\end{array}\right), \ \ \mathbf Q_3= \left(\begin{array}{cccc}
	0.036& 0.036& 0.453&0.475\\[2mm]
		0.036&0.036&0.475&0.453\\[2mm]
	0.037& 0.019& 0.472& 0.472\\[2mm]
	0.019& 0.037& 0.472& 0.472
\end{array}\right),
$$
$$p_3=(0.038, 0.038, 0.462, 0.462), \ \ q_3=(0.028, 0.028, 0.472, 0.472).
$$

\end{itemize}

The following is known (see \cite[p.55]{Ge}) as ergodic theorem for positive stochastic matrices.
\begin{thm}\label{to} Let $\mathbb P$ be a positive stochastic matrix and $\pi$
	the unique probability vector with $\pi \mathbb P=\pi$ (i.e.  $\pi$ is stationary distribution). Then
	$$\lim_{n\to \infty} x\mathbb P^n =\pi$$
	for all initial vector $x$.
\end{thm}

\subsection{Biological interpretations.}

Recall that a DNA is a configuration on a $\mathbb Z$-path.
By our construction only neighboring DNAs may interact.
The interaction is through an edge $l=\langle x, y\rangle\notin \mathbb Z$-path connecting two DNAs only when
configuration on this endpoints of the edge satisfy $\varphi(x)=\varphi(y)$, i.e., endpoints have the same color.

As a corollary of Theorem \ref{to} and above formulas of matrices and stationary distributions we obtain
the following biological interpretations: 

{\bf Case $\mu_1$:} {\bf No bubble coalescence:} With respect to measure $\mu_1$ of the BCI-DNA model on the Cayley tree of order 6 we have the following equilibrium state:
\begin{itemize}
\item two neighboring DNAs have junction of neighboring barrier zones with probability 0.976 (where states 1 and 2 seen with equal probability 0.488); they have junction of soft zones with probability 0.024 (where states 3 and 4 seen with equal probability 0.012). 

\item In a DNA the barrier zones seen with probability 0.968 (where states 1 and 2 have equal probability 0.484) and soft zones seen with probability 0.032 (where states 3 and 4 have equal probability 0.016).  
\end{itemize} 

{\bf Case $\mu$:} {\bf Domination of soft zone:} With respect to measure $\mu$ of the BCI-DNA model on the Cayley tree of order 6 we have the following equilibrium state:
\begin{itemize}
	\item two neighboring DNAs have junction of neighboring barrier zones with probability 0.306 (where states 1 and 2 have equal probability 0.153); they have junction of soft zones with probability 0.694 (where states 3 and 4 with probability 0.347). 
	
	\item In a DNA the barrier zones seen with probability 0.244 (where states 1 and 2 seen with probability 0.122) and soft zones seen with probability 0.756 (where states 3 and 4 have  probability 0.378).  
	\end{itemize} 

{\bf Case $\mu_2$:} {\bf Bubble coalescence:} With respect to measure $\mu_2$ of the BCI-DNA model on the Cayley tree of order 6 we have the following equilibrium state:
\begin{itemize}
	\item two neighboring DNAs have junction of neighboring barrier zones with probability 0.076 (where states 1 and 2 have probability 0.038); they have junction of soft zones with probability 0.924 (where state 3 and 4 have probability 0.462). 
	
	\item In a DNA the barrier zones seen with probability 0.056 (where states 1 and 2 with probability 0.028) and soft zones seen with probability 0.944 (where states 3 and 4 have probability 0.472).  
\end{itemize} 

\begin{rk} We note that above mentioned three equilibrium states of the BCI-DNA model are considered as coexistence of three phases: ``No bubble coalescence",  ``Dominated soft zone", ``Bubble coalescence". Since our measures are translation-invariant and each DNA has a countable set of neighbor DNAs, at the same temperature, each DNA interacts with several of its neighbors. DNAs having junctions (Holliday junction \cite{Rb}, \cite{Rp}) can be considered  as a
branched DNA. In case of coexistence more than one Gibbs measures, branches of a DNA can consist	different phases and different stationary states.
\end{rk}

\section{BCC-DNA model}

In this case boundary law equation (cf. with Theorem \ref{ei}) has the following form
\begin{equation}\label{zc}
	\begin{array}{ll}
z_{b,i,x}=\prod_{y\in S(x)}{1+\eta^{R(i,q)} z_{r,q,y}+\sum_{j=1}^{q-1}\left(z_{b,j,y}+\eta^{R(i,j)} z_{r,j,y}\right)\over 
		1+\eta^{R(q,q)} z_{r,q,y}+\sum_{j=1}^{q-1}\left( z_{b,j,y}+\eta^{R(q,j)}z_{r,j,y}\right)}, \ \ i\in \Phi_{q-1},\\[3mm]
	z_{r,i,x}=\prod_{y\in S(x)}{\zeta^{B(i,q)}+ z_{r,q,y}+\sum_{j=1}^{q-1}
		\left(\zeta^{B(i,j)}z_{b,j,y}+ z_{r,j,y}\right)\over 
		1+\eta^{R(q,q)} z_{r,q,y}+\sum_{j=1}^{q-1}\left( z_{b,j,y}+\eta^{R(q,j)}z_{r,j,y}\right)}, \ \ i\in \Phi_q,
	\end{array}
		\end{equation}
	where $\zeta=e^{J_b\beta}$ and $\eta=e^{J_r\beta}$. 

For functions (\ref{fc}), assuming that unknown functions do not depend on vertices of tree, for $q=2$ reduce system (\ref{zc}) to the following 
\begin{equation}\label{tic}
	\begin{array}{lll}
	u=\left({1+u+v+\eta w\over 1+u+\eta v+w}\right)^k,\\[2mm]
    v=\left({\zeta+u+v+w\over 1+u+\eta v+w}\right)^k,\\[2mm]
    w=\left({1+\zeta u+v+w\over 1+u+\eta v+w}\right)^k.
\end{array}
	\end{equation}
Consider this system as fixed point equation for the operator $G$ defined by 
\begin{equation}\label{tico}
G:\ \	\begin{array}{lll}
		u'=\left({1+u+v+\eta w\over 1+u+\eta v+w}\right)^k,\\[2mm]
		v'=\left({\zeta+u+v+w\over 1+u+\eta v+w}\right)^k,\\[2mm]
		w'=\left({1+\zeta u+v+w\over 1+u+\eta v+w}\right)^k.
	\end{array}
\end{equation}

Denote by ${\rm Fix}(G)$ the set of all fixed points of $G$ and 
$$L=\{(u,v,w)\in \mathbb R^3_+: u=1, v=w\}.$$
\begin{lemma}\label{in} If $\zeta<1$, $\eta>1$ then for any $k\geq 1$ we have ${\rm Fix}(G)\subset L$.
\end{lemma}
\begin{proof} Substracting 1 from the both sides of the first equation of (\ref{tic}) and substracting its second and third equations we get
	\begin{equation}\label{1v}
\left\{	\begin{array}{ll}
		(u-1)+P(\eta-1)(v-w)=0\\[2mm]
		Q(\zeta-1)(u-1)+(v-w)=0,
		\end{array}\right.
	\end{equation} 
where 
$$P={\sum_{i=1}^kA^{k-i}\over 1+u+\eta v+w}>0, \ \ Q={\sum_{i=1}^kB^{k-i}C^{i-1}\over 1+u+\eta v+w}>0$$
with
$$A={1+u+v+\eta w\over 1+u+\eta v+w}, \ \ B={\zeta+u+v+w\over 1+u+\eta v+w}, \ \ C={1+\zeta u+v+w\over 1+u+\eta v+w}.$$
From the second equation of (\ref{1v}) we get 
$v-w=Q(1-\zeta)(u-1)$ and substituting this to the first equation we get 
$$(u-1)[1+PQ(1-\zeta)(\eta-1)]=0.$$

This equality holds only for $u=1$, because $PQ>0$, $\zeta<1$. $\eta>1$. For $u=1$ from the second equation of (\ref{1v}) one gets $v=w$. 
\end{proof}

Thus all fixed points of the operator $G$ belong to $L$. 
Let us find fixed points of the operator $F$ on $L$. Then the system (\ref{tic}) is reduced to 
\begin{equation}\label{V}
	v=\left({\zeta+1+2v\over 2+(1+\eta)v}\right)^k.
	\end{equation}
\begin{lemma}\label{yag} For any $k\geq 1$, $\zeta\in (0,1)$, $\eta>1$ the equation (\ref{V}) has unique positive solution.
\end{lemma}
\begin{proof} Introduce  
\begin{equation}\label{abx}a={1+\zeta\over 2}<1, \ \ b={1+\eta\over 2}>1, \ \ x=\sqrt[k]{v},
	\end{equation}
and rewrite (\ref{V}) as
\begin{equation}\label{Vy}
	x={a+x^k\over 1+bx^k}.
\end{equation}
By (\ref{abx}) it is easy to see that if $x$ is a positive solution to (\ref{Vy}) then $0<x<1$. 

Rewrite (\ref{Vy}) as
$$g(x):=bx^{k+1}-x^k+x-a=0.$$
We show that $g(x)$ has unique root in $(0,1)$.
Note that $g(0)=-a<0$, $g(1)=b-a>0$. Thus there is at least one root  $x_*\in (0,1)$ of $g(x)$. To show uniqueness of $x_*$ it suffices to show that $g(x)$ is increasing on $(0,1)$, i.e.,
$$g'(x)=b(k+1)x^k-kx^{k-1}+1>0, \ \ \forall x\in (0,1) .$$ 
We have $g'(0)=1>0$, $g'(1)=b(k+1)-k+1>0$. It remains to show that minimum of $g'(x)$ is also positive. We have 
$$g''(x)=b(k+1)kx^{k-1}-k(k-1)x^{k-2}=0 \ \ \Leftrightarrow \ \ 
x=\hat x:={k-1\over b(k+1)}\in (0,1).$$
Consequently,
$$g'(\hat x)=1-\left({k-1\over b(k+1)}\right)^{k-1}>0.$$
Hence $g'(x)>0$ for any $x\in [0,1]$ and therefore $g(x)$ is increasing in this interval. 
\end{proof}

Denote by $\nu$ the translation invariant Gibbs measure which corresponds to unique solution $(1, x_*^k, x_*^k)$. 

\begin{rk}  We do not have explicit formula for solution $v=x_*^k$ if $k\geq 3$.  We only know its existence and uniqueness. But for small values of $k$ one can find the solution. For example, if $k=1$ then $v=x_*=\sqrt{{1+\zeta\over 1+\eta}}$. To give biological meaning of our measure $\nu$, below,  for $k=6$, and fixed parameters as in (\ref{tze}) we numerically find the unique solution of (\ref{V}).   This will be also nice to compare BCI-DNA and DCC-DNA models for the same parameters. 
\end{rk}
Summarizing Lemma \ref{in} and Lemma \ref{yag} we obtain the main result of this section:

\begin{thm} For the BCC-DNA model on the Cayley tree of order $k\geq 1$ if $q=2$, and $\zeta\in (0,1)$, $\eta>1$ then there exists unique translation-invariant Gibbs measure $\nu$.
\end{thm} 

Let $(u, v, w)$ be a solution to (\ref{tic}), which by Lemma \ref{in} has the form $(1, v, v)$. The Markov chain (Gibbs measure) corresponding to this solution is defined by the following  matrix 

	$$\mathcal P=\left(\begin{array}{cccc}
	{1\over Y}& {1\over Y}& {v\over Y}& {\eta v\over Y}\\[2mm]
	{ 1\over Y}& {1\over Y}& {\eta v\over Y}& {v\over Y}\\[2mm]
	{ 1\over Z}& {\zeta\over Z}& { v\over Z}& {v\over Z}\\[2mm]
	{\zeta  \over Z}& {1\over Z}& { v\over Z}& {v\over Z}
\end{array}\right),
$$
where $Y=2+(1+\eta)v$,  $Z=1+\zeta+ 2v$.

For concrete parameters $k=6$ and (\ref{tze}) we have unique positive solution to (\ref{V}): $v=0.256135892$. 
For this solution the matrix $\mathcal P$ has the following 
stationary probability vector: $p=(0.391, 0.391, 0.109, 0.109)$.
Then for the corresponding measure $\nu$ we have the following 

 {\bf Domination of barrier zone:} With respect to unique measure $\nu$ of the BCC-DNA model on the Cayley tree of order 6 we have the following equilibrium state:

{\it In a DNA the barrier zones seen with probability 0.782 (where states 1 and 2 seen with probability 0.391) and soft zones seen with probability 0.218 (where states 3 and 4 have  probability 0.109).} 
\begin{rk}
	The last result shows that bubble coalescence does not hold (with high probability) if one considers a Cayley tree as one molecule of condensed DNA with the same parameters as in BCI-DNA model.   But previous section showed that for BCI-DNA model (with the same parameters) the bubble coalescence holds.   
\end{rk}

\section*{Acknowledgements}

The author thanks Institut des Hautes \'Etudes Scientifiques (IHES), Bures-sur-Yvette, France for support of his visit to IHES.  His  work was partially supported by a grant from the IMU-CDC and the fundamental project (number: F-FA-2021-425)  of The Ministry of Innovative Development of the Republic of Uzbekistan.

\section*{Statements and Declarations}
	
{\bf	Conflict of interest statement:} 
The author states that there is no conflict of interest.

\section*{Data availability statements}
The datasets generated during and/or analyzed during the current study are available from the corresponding author on reasonable request.

\end{document}